\documentclass[twocolumn,10pt]{IEEEtran}
\usepackage{amsmath}
\usepackage{subfigure}
\usepackage{epsf}
\usepackage{url}
\usepackage[usenames]{color}
\usepackage{graphicx}
\newtheorem{proposition}{Proposition}
\newtheorem{lemma}{Lemma}
\centerfigcaptionstrue

\begin{document}

\title{The Capacity of Ad hoc Networks under Random Packet Losses}
\author{
Vivek P.~Mhatre,
Catherine P.~Rosenberg,
Ravi R.~Mazumdar\thanks{Vivek P.~Mhatre is with Motorola Inc., Arlington Heights, IL, USA. Email: mhatre@gmail.com,
Catherine P.~Rosenberg and Ravi R.~ Mazumdar are with the Department of Electrical and Computer 
Engineering, University of Waterloo, Canada. Email: \{cath,mazum\}@ece.uwaterloo.ca.}
}


\maketitle

\begin{abstract}
We consider the problem of determining asymptotic bounds on the capacity of a random ad hoc network\footnote{Preliminary versions of this work appeared in  \cite{isit,preprint}.}.
Previous approaches assumed a 
link layer model in which if a transmitter-receiver pair can communicate with each other, i.e., the
Signal to Interference and Noise Ratio (SINR) is above a certain threshold, then every transmitted packet
is received error-free by the receiver thereby. Using this model, the per node capacity of the network
was shown to be $\Theta\left(\frac{1}{\sqrt{n\log{n}}}\right)$.
In reality, for any finite link SINR, there is a non-zero probability of erroneous reception of the packet.
We show that in a large network, as the packet
travels an asymptotically large number of hops from source to destination, the cumulative impact of
packet losses over intermediate links results in a per-node throughput of only $O\left(\frac{1}{n}\right)$.
We then propose a new scheduling scheme to counter this effect.
The proposed scheme provides tight guarantees on end-to-end packet loss probability,
and improves the per-node throughput to $\Omega\left(\frac{1}{\sqrt{n}
\left({\log{n}}\right)^{\frac{\alpha{{+2}}}{2(\alpha-2)}}}\right)$ where $\alpha>2$ is the path loss exponent.
\end{abstract}

\section*{Keywords}
Ad hoc networks, Capacity, SINR, Interference, Scheduling

\section{Introduction}
The problem of the capacity of wireless ad hoc networks was first
analyzed by Gupta and Kumar in their seminal work \cite{gupta}. 
This work was followed by several studies on the capacity of wireless ad hoc networks \cite{liang}, 
\cite{gaurav}, among others.
The authors in \cite{gupta} derive asymptotic
bounds on the capacity of a random ad hoc network
in which nodes are deployed randomly and uniformly over the surface of a
sphere of unit area. Each node picks a
random node as its destination node, and sends packets to that node by using
multi-hop communication. All the nodes use a common transmit power level.
The authors show that 
each node can achieve a throughput of $\Theta\left(
{\frac{1}{\sqrt{n\log{n}}}}  \right)$ packets per second.
They also provide a scheduling and routing strategy that achieves this throughput.

The authors in \cite{gupta} assume a link layer model in which, if the Signal to
Interference and Noise Ratio (SINR) at the receiver is greater than a certain threshold then
the packet is received successfully by the receiver.
In practical wireless networks, for a given modulation and coding scheme, and for a fixed block length,
as long as there is some noise and interference, i.e., as long as
the SINR of a link is finite, there is a certain non-zero probability of packet error.
Hence, the hypothesis of perfect packet reception when the link SINR is above a certain threshold can be realized only when
one of the following is assumed: (i) infinite block lengths, or (ii) infinite number of retransmissions, i.e., retransmit the packet
until it gets through. Both these approaches lead to unbounded delays, and hence under either (i) or (ii), a delay-capacity tradeoff
study should be undertaken instead of a capacity study.
In practical systems in order to keep the overall delay under tabs, the block length, as well as the maximum number of
retransmission attempts (in case retransmissions are employed) are pre-determined, and fixed.
Hence, the threshold-based packet reception model used in \cite{gupta} is a reasonable choice for successful packet
reception in a single hop network such as a
cellular network. However, we argue that it needs to be refined when applied to a multi-hop network.
In an ad hoc network, each packet traverses 
{\em multiple hops}. 
When a packet is relayed over a large number of links, each of which being likely to drop the packet
with a certain probability (no matter how small it is), 
the end-to-end throughput can degrade significantly due to the {\em{cumulative}} packet error probability.

More generally, in this paper, we 
show that when studying capacity scaling problems, the underlying hypotheses have a 
significant impact on the results. In particular, we show that when packet losses over links are taken into account,
we get strikingly different results thereby showing the sensitivity of capacity studies
to the underlying hypotheses.

For a random ad hoc network, for a broad range of routing and scheduling schemes including the one
proposed in \cite{gupta},
we show that the cumulative impact of  per link packet loss results in a much lower per-node throughput
of $O\left(\frac{1}{n}\right)$ instead of $\Theta\left({\frac{1}{\sqrt{n \log{n}}}}\right)$. 
%
In order to counter the above throughput reduction due to cumulative packet error effect,
we propose a new scheduling policy that uses lower spatial reuse to
reduce interference, and thereby improves the SINR of each link. 
The proposed scheduling policy improves the end-to-end packet success probability, and
results in a per-node throughput of $\Omega\left(\frac{1}{\sqrt{n}
\left({\log{n}}\right)^{\frac{\alpha{{+2}}}{2(\alpha-2)}}}\right)$,
where $\alpha>2$ is the path loss exponent.


The rest of the paper is organized as follows. 
We motivate the problem, and outline our approach in
Section \ref{approach}. The analytical models used for the physical and the
network layer are presented in Section \ref{prob}.
We present results on throughput reduction due to cumulative packet loss
in Section \ref{secpi1} and Section \ref{secpi1rtng}.
We present our new scheduling policy, and obtain the corresponding
capacity results in Section \ref{secpi2}.
We discuss some of the related work in Section \ref{related}.
Finally, we present our conclusion and future directions in Section \ref{conclude}.


\section{Motivation and Approach}
\label{approach}
The arguments in this section are only meant to provide intuitive insights about the problem and our approach. 
We provide precisely proved results for all the arguments in Section \ref{prob}.
In \cite{gupta}, the authors study the problem of determining the capacity of random ad hoc networks.
They propose a routing and scheduling scheme to achieve a per-node throughput
of $\Theta\left({\frac{1}{\sqrt{n \log n}}}\right)$. 
It can be shown that for this scheme a packet traverses $\Theta\left( \sqrt{\frac{n}{\log n}}\right)$ intermediate hops
on its way from a source node to its destination node.
Thus asymptotically, the number of hops that a packet has to travel from source to destination goes
to infinity as $n$ scales.
The authors assume that if the SINR over a hop is at least $\beta$,
the packet transmission is successful. Let us call this the ideal link model.
The authors then use a scheduling scheme which ensures an SINR of at least $\beta$ over all the
hops. Thus, under the ideal link model, and the proposed scheduling policy, {\em every}
packet transmission is successful.

Now consider another link model in which a link is reliable with a probability $p<1$, i.e.,
each packet transmitted on the link is received successfully at the receiver with a probability of $p$.
Let us call this the probabilistic lossy link model.
For simplicity, assume that over each hop, if a packet is not received successfully
by the receiver, the transmitter does not retransmit the packet, and the packet is lost
(the generalization to handle retransmissions is addressed subsequently).
Also assume that we use the same routing and scheduling policy as used above for the ideal link model.
Since there are  $\Theta\left( \sqrt{\frac{n}{\log n}}\right)$ hops from source to destination,
the probability that the packet reaches its destination scales as $p^{\sqrt{\frac{n}{\log n}}}$.
It is easy to show that this quantity is $O\left(\sqrt{\frac{\log n}{n}}\right)$, since $p^m$ is
$O\left(\frac{1}{m}\right)$ as $m$ tends to infinity.
Let us consider two identical networks, one with the ideal link model, and the other with the 
probabilistic lossy link model, and assume that for both of them, we use the same routing and scheduling
algorithm as outlined in \cite{gupta}. The only difference between the two networks is that while in
the former case, all the packets injected by the source into the network reach the destination, in
the latter case, only a fraction of the injected packets reach the destination.
Hence, under the probabilistic lossy link model, the achievable throughput is the product of the achievable
throughput of the ideal link model, $\Theta\left({\frac{1}{\sqrt{n \log n}}}\right)$, and
an end-to-end packet delivery probability term that is $O\left(\sqrt{\frac{\log n}{n}}\right)$.
This results in an end-to-end throughput of $O\left(\frac{1}{n}\right)$ for the probabilistic lossy
link model.

A realistic link model lies somewhere between the ideal link model and the probabilistic lossy
link model. In a realistic link model, the probability of link reliability, $p$, is a continuous
function of the SINR
of the link. The SINR in turn depends on factors such as transmitter-receiver separation, power of
interference from simultaneous transmissions, and noise power. These factors are different for
different links along a path. We first show that despite ensuring a minimum SINR of $\beta$, there is a non-zero
packet loss probability over each link.
We then show that since the scheduling approach in \cite{gupta} does not take into account the
cumulative packet losses along a path,
it results in a substantially lower throughput ($O\left(\frac{1}{n}\right)$ 
instead of $\Theta\left(\frac{1}{\sqrt{n\log{n}}}\right)$).
We then propose a new scheduling policy under which the throughput performance is $\Theta\left(\frac{1}{\sqrt{n}
\left({\log{n}}\right)^{\frac{\alpha{{+2}}}{2(\alpha-2)}}}\right)$ even in
the presence of packet losses.
%

\section{Model for Random Networks}
\label{prob}
We use the standard ordering notations for $O(.), o(.), \Omega(.), \omega(.), \Theta(.)$.

\noindent{\bf{Scaling model:}}
We assume that $n$ nodes are randomly and uniformly deployed
over a sphere $S^2$ of area $n$ so that the node density is kept constant. Under this model, the far-field
assumption can be employed in the path loss model \cite{scaling}.
The nodes are assumed to be stationary. Each node picks a
random node as its destination node, and sends packets to this destination node. 
Each node has a common transmit power level $P$, and uses intermediate nodes as relays to reach
its destination. 
We assume that as $n$ scales, the modulation and coding scheme, as well as the block length remain fixed.
Furthermore, we assume that a common frequency chunk of $W$ Hz is available to all the nodes.
This models the scenario in which although the network size increases,
the underlying hardware and software capabilities of the nodes, as well as the available spectrum are unchanged.

\noindent{\bf{Physical layer model:}}
The SINR at a receiver node $j$ when transmitter node $i$ is transmitting a packet to $j$ is as follows.
\begin{equation}
\label{model}
\text{SINR} = \frac{ \frac{P}{{|X_i - X_j|}^{\alpha}}}{ N + \sum\limits_{k\in T \mbox{, } k\neq i}\frac{P}{{|X_k - X_j|}^{\alpha}}}
\mbox{, }
\end{equation}
where $N$ is the noise power, $X_i$ is the position vector of node $i$,
$T$ is the set of all the nodes that are transmitting simultaneously with node $i$, and $\alpha
> 2$ is the path loss exponent. 
As in \cite{gupta}, we assume that the transmit power
of all the nodes can be scaled with $n$ so that the effects of receiver noise can be ignored,
and the network becomes interference dominated. We therefore ignore the term $N$ in the denominator of SINR henceforth.
Two nodes can communicate with each other if the SINR at the receiver node is greater than a certain threshold, say $\beta$.
In \cite{gupta}
the authors assume that if two nodes can communicate with each other, then all the transmitted packets are received error-free.
%
Thus, if $\hat{\phi}(\cdot)$ is the mapping between the SINR and the 
probability of successful packet reception, then the model used in \cite{gupta} is equivalent to the following.
\begin{equation}
\hat{\phi}(\text{SINR}) =
\begin{cases}
1 & \text{ if $\text{SINR} \geq \beta$, }\\
0 & \text{ otherwise.}
\end{cases}
\label{phi}
\end{equation}

{\em{The key contribution of our work is to derive capacity results under the
following more accurate physical layer model.}}
We assume that the probability of a error-free packet reception is a continuous increasing
function of SINR, $\phi(\cdot)$
that approaches unity as the SINR approaches infinity. Hence, for every finite SINR value, 
there is a non-zero probability that the packet is received in error which amounts to packet loss.
The results carry over to the case in which
the number of retransmissions allowed is finite but fixed, i.e., do not depend on $n$.

\noindent{\bf{Interference model:}}
We assume that the interference observed by a packet over different hops along its path from source to destination is independent
across all the hops.
Note that the set of interferers active during the transmission of a packet over different
hops, and the actual packets that the interfering nodes
transmit during those slots are unlikely to be correlated as the packet moves along its path.
Furthermore, the interfering signals themselves, i.e., the sequences of bits in the interfering packets
can be assumed to be independent of each other. Hence the assumption of
independence of interference across multiple hops.

\noindent{\bf{Network connectivity and routing:}}
In order to ensure connectivity, we use the connectivity result from \cite{guptaconn}.
The surface of the sphere is covered by a Voronoi tessellation in such a way that each Voronoi
cell $V$ can be enclosed inside a circle of radius $2\rho_n$, and each circle encloses a circle of
radius $\rho_n$ (all the distances are measured along the surface of $S^2$). We use such a tessellation
to ensure uniform cell size. 
In \cite{guptaconn}, Gupta
and Kumar have derived necessary and sufficient conditions for asymptotic connectivity of a random ad hoc network for
the case of unit network area.
By using results from \cite{guptaconn} in \cite{gupta}, the authors choose $\rho_n$ to be the radius of a circle
of area $\frac{100\log{n}}{n}$ on a sphere of unit area. With this choice of $\rho_n$, and with the communication
radius of each node chosen to be $8\rho_n$, the authors show that all the nodes in the network are
connected with probability approaching unity as $n$ approaches infinity. 
Since we consider a scaled up version of the network deployment in which the network area is scaled by a factor $n$,
the same connectivity results in \cite{guptaconn} are applicable to our model with a normalization factor of
$\sqrt{n}$ for all the distances, and a normalization factor of $n$ for all the areas. Hence, under our model,
$\rho_n$ is chosen to be the radius of a circle of area $100\log{n}$, thereby implying that
$\rho_n=\Theta\left(\sqrt{\log{n}}\right)$.
We assume the straight line routing policy from \cite{gupta} in which 
packets are routed along straight line paths between 
source-destination pairs, i.e., every cell that intersects the straight
line joining a source-destination pair, relays the packets of that pair. The
generalization to handle arbitrary routing is presented subsequently.
\noindent{\bf{Scheduling Policy:}}
We first consider the scheduling policy proposed in \cite{gupta}, and obtain capacity results under this policy in Section
\ref{secpi1}. We refer to this
policy as $\pi_1$.
This scheduling policy uses a vertex-coloring algorithm to guarantee that each cell gets
a transmission opportunity at least once every $K$ slots, where $K$ is a constant that is
independent of $n$. In other words, the length of the schedule is bounded even as $n$ scales.
In Section \ref{secpi2}, we then propose a new improved scheduling policy, $\pi_2$ that yields
better capacity under our physical layer model. The length of the schedule under $\pi_2$, $K_n$ grows as 
a function of $n$,


\section{{Capacity under Scheduling Policy $\pi_1$ and continuous $\phi(\cdot)$}}
\label{secpi1}
Recall that we assume that $\phi(.)$ approaches unity continuously as the SINR goes to infinity.
Under this model, we show that the scheduling policy $\pi_1$ used in
\cite{gupta} results in a network capacity of $O\left(\frac{1}{n}\right)$.

Let $L_i$ be the line segment along the surface of $S^2$ that connects the $i$-th
source-destination pair (henceforth referred to as the $i$-th connection).
We also use $L_i$ to denote the length of the line segment
joining the $i$-th source-destination pair. As per the straight line routing scheme, the packets of the $i$-ith
connection are relayed hop-by-hop by every cell which intersects line $L_i$. Over each hop, any node
in the relaying cell may forward the packet. The scheduling algorithm, $\pi_1$, and the uniform cell sizes
ensure that communication between any two nodes in the neighboring cells is possible by guaranteeing
that the SINR at the receiving node is greater than or equal to $\beta$ (see \cite{gupta} for more
details). The following lemmas are presented without proofs.

\begin{lemma}
\label{lemmaH}
For the straight line routing scheme, the number of hops $H_i$ for connection $i$ is
$\Theta\left( \frac{L_i}{\rho_n} \right)$. More precisely,
\begin{equation*}
\frac{1}{8} \frac{L_i}{\rho_n} \leq H_i \leq \frac{16}{\pi} \frac{L_i}{\rho_n}\mbox{.}
\end{equation*}
\end{lemma}
\begin{proof}
See Appendix B.
\end{proof}
%
\begin{lemma}
\label{lemmaHh}
Fix $t$ such that $0 < t < 1$. For connection $i$, out of $H_i$ total hops, let $h_i$ hops be such that each of
these hops covers a distance of less than $t \rho_n$. Then
\begin{equation*}
H_i - h_i \geq \frac{L_i}{\rho_n} \left(\frac{1 - \frac{16 t}{\pi}}{8 - t} \right).
\end{equation*}
Thus, for the above $H_i - h_i$ hops, the signal received at the receiver is at the most $P (t
\rho_n)^{-\alpha}$.
\end{lemma}
\begin{proof}
See Appendix B.
\end{proof}

In \cite{gupta}, 
scheduling policy $\pi_1$ was proposed to guarantee an SINR of at least $\beta$ at the receiver of
every scheduled transmission. 
This scheduling policy corresponds to a graph coloring problem, and it was
shown that the maximum number of colors required to color all the cells 
is upper bounded by $1 + c_1$, where $c_1$ is a fixed constant that is independent of $n$ (see Lemma 4.4 in \cite{gupta}).
Using this scheduling scheme, each cell gets a transmission opportunity at least once every $1 + c_1$
time slots.
The following lemma shows that under such a scheduling strategy,
except for a small fraction of hops, all the
remaining hops of a connection receive a certain minimum amount of interference from other
simultaneous transmissions.
\begin{lemma}
\label{lemmaM}
Fix $M>9$.
Let $N_i$ be the number of hops of connection $i$ such that
there is no simultaneous interfering transmission within a circle of radius $(M+8) \rho_n$ around
the receivers of those hops. Then,
\begin{equation*}
N_i \leq \frac{L_i}{\rho_n} \left( \frac{2(1 + c_1)}{M}\right)\mbox{,}
\end{equation*}
where $c_1$ is the constant from Lemma 4.4 in \cite{gupta}.
\end{lemma}
\begin{proof}
See Appendix B.
\end{proof}

Using the above lemma, we have the following important result.
\begin{proposition}
\label{propphi}
There exist fixed constants $t_0$ and $M_0$, that do not depend on $n$, such that for at least
$L_i/16\rho_n$ hops of connection $i$, the SINR is less than a fixed constant $\beta_0$ where
\begin{equation*}
\beta_0    = \left( \frac{M_0 + 8}{t_0} \right)^{\alpha}.
\end{equation*}
Since $\phi(\cdot)$ is continuous, and the SINR is upper bounded by a fixed constant $\beta_0$, the
probability of successful packet reception for these hops is also upper bounded by a fixed
constant $\phi(\beta_0) < 1$.
\end{proposition}
\begin{proof}
See Appendix B.
\end{proof}

As the number of hops between source and destination scales to infinity, and as per the above result,
a fixed fraction of those hops have a certain minimum non-zero packet loss probability. Hence, we can show that
the expected end-to-end throughput of a node after accounting for the packet losses is then given
by the following result.
\begin{proposition}
The expected end-to-end throughput of each node, ${\bf E} [\Lambda_n]$ is upper bounded by
\begin{equation}
{\bf E} [\Lambda_n ] < \lambda_n c_{0} \frac{\log{n}}{n}, 
\label{Lambdainter}
\end{equation}
where $\lambda_n$ is the rate at which each node injects packets into the network,
and $\rho_n$ is the size of the cell (each cell in contained in a circle of radius $2\rho_n$,
and contains a circle of radius $\rho_n$).
\label{propinter}
\end{proposition}
\begin{proof}
See Appendix B.
\end{proof}
When $n$ nodes are distributed randomly and uniformly over a region, and the region consists of disjoint convex cells
(each of area at least $100\log{n}$) such
that the probability of a node falling in a given cell is at least $\frac{\log{n}}{n}$, then with high probability,
the number of nodes in each cell is at least $50\log{n}$ (see Lemma 4.8, \cite{gupta}).
More precisely, there exists a sequence $\delta_n \rightarrow 0$, such that
\begin{align*}
\mbox{Prob} &\left\{\mbox{Number of nodes in cell $V$ $\geq$ $50\log n$, for every}\right.\\
&\left.\mbox{cell $V$ in the tessellation}\right\} > 1 - \delta_n.
\end{align*}
Thus we have the following Lemma for the above choice of $\rho_n$.
\begin{lemma}
\label{lemma50logn}
If $\lambda_n $ is the rate in packets/slot at which {\em every}\, source node injects packets in the network under scheduling
policy $\pi_1$, then with high probability,
\begin{equation*}
\lambda_n \leq \frac{1}{50 \log n}.
\end{equation*}
\end{lemma}
\vspace{0.2cm}

\begin{proof}
The scheduling algorithm proposed in \cite{gupta} guarantees each cell a time slot at least once every $1+ c_1$ slots. However,
since each cell contains at least $50 \log n$ nodes with high probability, even if each node were to
transmit only its own packets, and not relay packets, it would still get no more than one
transmission opportunity every $50 \log n$ slots.
Hence, the rate at which a node injects its own packets into the network can never be more than
$1/50\log{n}$.
\end{proof}
\vspace{0.2cm}

Using the above Lemma to bound $\lambda_n $ in \eqref{Lambdainter}, and for an appropriately defined constant
$c_1$, we get
\begin{equation}
{\bf E} [\Lambda_n] \leq \frac{c_1}{n}.
\label{Lambdalast}
\end{equation}
Thus we have proved the following main result.
\begin{proposition}
\label{propmain1}
Under scheduling policy $\pi_1$, continuous $\phi(\cdot)$, and with $\rho_n$ chosen to be the
radius of a disk of area $100 \log{n}$, the per-node throughput that can be achieved
is $O(\frac{1}{n})$.
\end{proposition}
\begin{proof}
See Appendix B.
\end{proof}


{\bf {Remark:}} The above bound on the throughput also holds for a broad class of
routing schemes, i.e., Proposition \ref{propmain1} holds even when the assumption of straight-line routing is relaxed,
(see Appendix A for details).

\section{Capacity under a new Scheduling Policy $\pi_2$, and continuous $\phi(\cdot)$}
\label{secpi2}
In this section, we show that if we
reduce the extent of spatial reuse via scheduling, then we can bound the end-to-end packet
loss probability, {\em{not just}}
the per link packet loss probability. This restricts the throughput reduction
due to cumulative packet loss.

\begin{lemma}
For a given $\epsilon>0$, if the packet loss probability over each link is upper bounded by
$\frac{\sqrt{\pi}\rho_n \epsilon}{8\sqrt{n}}$,
then under straight line routing policy, the cumulative packet loss probability for
each connection is upper bounded by $\epsilon$.
\label{unionbnd}
\end{lemma}
\begin{proof}
See Appendix B for proof.
\end{proof}
Note that the above implies that as $n$ scales, in order to counter cumulative
packet loss, all the links have to be progressively more reliable.
We assume that $\rho_n=10\sqrt{\log{n}}$ so that
connectivity is achieved with high probability \cite{gupta}.

\begin{lemma}
For a given $\epsilon>0$, if a scheduling policy ensures that each transmission has an SINR 
of at least {{{$\beta_n\geq {\log{n}}$}}}, and if the network is large enough, i.e., 
$n > \frac{4}{5\sqrt{\pi}\epsilon}$, then each connection has an end-to-end cumulative
packet loss probability of no more than $\epsilon$.
\label{schedreq}
\end{lemma}
\begin{proof}
See Appendix B for proof.
\end{proof}

\begin{lemma}
For each scheduled transmission of a cell an SINR of at least
{{$\beta_n={\log{n}}$}} is
guaranteed if all the cells that are within a 
distance $R_n$ of the given cell are silent, where $R_n$ is given by:
\begin{equation}
R_n = 
20\sqrt{{\log{n}}} \left(  +
4\left(\frac{256{{{\log{n}}}}}{(\alpha-2)}\right)^{\frac{1}{(\alpha-2)}}\right)
\label{ropt}
\end{equation}
\label{pisched}
\end{lemma}
\begin{proof}
See Appendix B for proof.
\end{proof}

Concurrent transmissions
that are at least $R_n$ distance apart can then be scheduled using the following policy.

\noindent{\bf{Scheduling policy }}{\pmb{$\pi_2$:}} 
Consider the graph in which each cell is represented by a vertex. Two vertices of the graph
have an edge between them if and only  if their corresponding cells are within a distance of
$R_n$ from each other. Since the area of each cell is lower bounded by
${50\pi\log{n}}$, the number of cells within a
distance $R_n$ of a cell, and correspondingly the maximum degree of a vertex in
the corresponding graph is upper bounded by $V_n$ given by
\begin{equation}
V_n
= 8\left(1+4\left(\frac{256{{{\log{n}}}}}{(\alpha-2)}\right)^{\frac{1}{\alpha-2}}\right)^2.
\label{vneq}
\end{equation}
Since a graph in which the maximum degree of a node is $k$ can be vertex colored using no more than
$k+1$ colors in polynomial time \cite{bondy}, we use this algorithm to
vertex color the graph. Thus, the required number of colors is no more than $K_n =
V_n+1$, which is a function of $n$ unlike $\pi_1$ where the number of colors is a constant independent of $n$.

Under scheduling policy $\pi_2$ we showed that the received interference power is bounded.
Indeed  it can be seen that for a path loss
exponent $\alpha > 2$ all $r$-th moments of the interference for $r \geq 2$ are bounded. 
Moreover, since all moments exist, the interference satisfies the large deviations principle, i.e., a
moment generating function of the interference is well defined, and hence the distribution function of interference has exponentially decaying
tail.
Hence a ``good'' rate function, $I(x)$, exists by Cramer's theorem, i.e. $P(\text{Interference} > x) \leq e^{-I(x)}$ where 
$\lim_{x\to\infty} \frac{I(x)}{x} = \infty$.
This implies that the tail distribution of the interference has  at least exponential decay.
This property is sufficient to conclude that the packet loss probability will also 
have a good rate function associated with it since the function of the SINR is smooth by the contraction principle.
From this we can conclude that there is
at least an exponential decay of the packet loss probability which is all that we need to establish our results\footnote{Note
that if we can show that the moment generating function, $M(t)$, of the interference is such that $\log{M(t)}$ is differentiable for all $t$,
the rate function is strictly convex \cite{ellis}, and the tail fall-off of $e^{-x^2}$ can be obtained. In this case, we only need to ensure an SINR of
$\sqrt{\log{n}}$ instead of $\log{n}$, and the capacity bound can be improved. However, in order to prove that $\log{M(t)}$ is differentiable 
it is necessary to explicitly compute all the $r$-th moments of the interference term.}.

\begin{proposition}
For a given fixed $\epsilon>0$
with scheduling policy $\pi_2$, for a large network, i.e., if $n > \frac{4}{5\sqrt{\pi}\epsilon}$,
each connection has an end-to-end packet loss probability of less than $\epsilon$.
\label{pi2guar}
\end{proposition}
\begin{proof}
The proof follows follows from exponential tail decay.
See Appendix B for proof
\end{proof}

We thus have the following main result.
\begin{proposition}
When random packet losses over intermediate links are taken into account, the
per-node throughput in an ad hoc network is $\Omega\left(\frac{1}{\sqrt{n}
\left({\log{n}}\right)^{\frac{\alpha{{+2}}}{2(\alpha-2)}}}\right)$, and the bound 
can be achieved under Scheduling Policy $\pi_2$.
\end{proposition}
\begin{proof}
If each connection generates packets at a rate of $\lambda_n $, then from Lemma
4.12 in \cite{gupta}, the total load of transmitting packets on each cell is
upper bounded by $c_5\lambda_n \sqrt{n\log{n}}$ with high probability.
When the cell gets an opportunity to transmit, it uses the entire bandwidth to
transmit at a rate of $W$ bits per second. Thus, with high probability, a rate
of $\lambda_n $ can be scheduled if the
following holds.
\begin{align}
&c_5\lambda_n \sqrt{n\log{n}} \leq \frac{W}{K_n}\nonumber\\
&\Rightarrow \lambda_n  \leq \frac{W}{c_5 \sqrt{n\log{n}} \left(a_0 + a_1 \left(\log{n}\right)^{\frac{1}{(\alpha-2)}} + a_2
\left(\log{n}\right)^{\frac{{{2}}}{\alpha-2}}\right)}
\end{align}
using $K_n=V_n+1$, and Eq.~\eqref{vneq}, and 
where $a_0, a_1$, and $a_2$ are constants that do not depend on $n$. 

The goodput after accounting for cumulative packet losses is
then given by:
\begin{align}
\text{\bf{E}} [ \Lambda_n ] &= \frac{(1-\epsilon)W}{c_5 \sqrt{n\log{n}} \left(a_0 + a_1 \left(\log{n}\right)^{\frac{1}{(\alpha-2)}} + a_2
\left(\log{n}\right)^{\frac{{{2}}}{\alpha-2}}\right)} \nonumber\\
&= \Omega\left( \frac{1}{\sqrt{n}
\left({\log{n}}\right)^{\frac{\alpha{{+2}}}{2(\alpha-2)}}}\right).
\end{align}
\end{proof}

%

\section{Related Work}
\label{related}
In their seminal work on the capacity of wireless networks \cite{gupta}, Gupta and Kumar
used a simplified link layer model in which each packet
reception is successful if the receiver has an SINR of at least $\beta$.
For this communication model and under the random network model, the authors propose a
routing and scheduling strategy, and show that a per-node throughput of
$\Theta \left( \frac{1}{\sqrt{n \log{n}}}\right)$ can be achieved.
The authors also consider another model of ad hoc networks called arbitrary networks.
In an arbitrary network, nodes are placed at arbitrary
locations in a region of fixed area, source-destination pairs are chosen arbitrarily, and each node can
transmit at an arbitrary power level. This is unlike a random ad hoc network where node are deployed
randomly and uniformly, source-destination pairs are chosen at random, and all the nodes transmit at
the same power level\footnote{The definitions of random networks and arbitrary networks have been taken from 
\cite{gupta}.}. For the arbitrary network model, the authors show that a per-node throughput of
$\Theta\left(\frac{1}{\sqrt{n}}\right)$ is achievable. 
However, for random networks, the routing and scheduling scheme proposed in
\cite{gupta} only achieves a throughput of $\Theta\left(\frac{1}{\sqrt{n\log{n}}}\right)$. 


The gap
between the arbitrary networks and the random networks is closed in \cite{franc}. In \cite{franc}, the
authors note that the requirement of connectivity with high probability in random networks as required
in \cite{gupta}, requires higher transmission power at all the nodes, and results in excessive
interference. This in turn lowers the throughput of random networks from $\Theta\left(\frac{1}{\sqrt{n}}\right)$ to
$\Theta\left(\frac{1}{\sqrt{n\log{n}}}\right)$. Motivated by this, the authors in \cite{franc}
propose using a backbone-based relaying scheme in which instead of ensuring connectivity with
probability one, they allow for a small fraction of nodes to be disconnected from the backbone. 
The nodes in the backbone are densely connected, and can communicate over each hop at a constant rate. 
Such a backbone traverses up to $\sqrt{n}$ hops. The nodes that are not a part of the backbone, send
their packets to the backbone using single hop communication. The authors then show that the
interference caused by these long range transmissions does not impact the traffic carrying capacity of
the backbone nodes. Finally, the authors show that the relaying load of the backbone determines 
the per-node throughput of such a scheme, and this results in an achievable per-node throughput of
$\Theta\left(\frac{1}{\sqrt{n}}\right)$. However, as in \cite{gupta}, the authors in \cite{franc} do not take into account
the impact of cumulative packet loss.
The approach of allowing a few disconnected nodes in the network has also been used in
\cite{dousse}. Capacity of ad hoc networks under cooperative relaying is studied in \cite{scaling} where instead
of decode and forward model, node are clustered in MIMO configurations.

In \cite{liang}, the authors discuss the limitations of the work in
\cite{gupta} by taking a network information theoretic approach. The
authors discuss how several cooperative strategies such
as interference cancellation, network coding, etc. could be used to
improve the throughput.
The authors also discuss how
the problem of determining the network capacity from an information
theoretic view-point is a difficult problem, since even the capacity of
a three node relay network is unknown. In Theorem 3.6 in \cite{liang},
the authors determine the same bound on the capacity of a random network as 
obtained in \cite{gupta}.
All the above works do not take into account the capacity reduction due to cumulative packet loss
which is the focus of our work.
In addition to the above works, the capacity of ad hoc networks has also been investigated under 
additional assumptions such as node mobility and
capacity-delay
trade-offs \cite{gaurav}, network information theoretic view of network
capacity \cite{liang}, etc.

\section{Conclusion}
\label{conclude}
Our results show that idealized modeling assumptions can lead to optimistic conclusions,
and hence for scaling techniques to provide useful guidelines it is important to understand the limitations of the
underlying assumptions.
The key observation in this paper is that
for a broad range of routing and scheduling schemes (including the one proposed in \cite{gupta}),
the number of hops of each connection scales to infinity with $n$.
In a realistic link layer model, for a given modulation and coding scheme, the
probability of packet loss over any given link is non-zero for finite SINR values.
With the above link layer model, we show that for a broad class of routing and scheduling policies
we cannot achieve a per-node throughput of more than $O \left( \frac{1}{n}\right)$ due to the
cumulative probability of packet loss over all the hops of the connection. 
However, it is possible to improve the link SINR by using a more conservative scheduling policy with a
lower spatial reuse. In this case the per-node throughput is $\Omega\left(\frac{1}{\sqrt{n}
\left({\log{n}}\right)^{\frac{\alpha{{+2}}}{2(\alpha-2)}}}\right)$. This bound is 
achievable under the proposed scheduling policy.

\section*{Acknowledgments}
This work was supported in part by the Natural Sciences and Engineering                                             
Research Council of Canada (NSERC), the National Science Foundation                                                 
(NSF), and a grant from Research in Motion (RIM).
We would also like to thank Prof. Ness Shroff (Ohio State University, USA) for
his valuable suggestions that greatly helped us improve the quality of this manuscript.

\section*{Appendix A: Capacity under Scheduling Policy $\pi_1$, continuous $\phi(\cdot)$ and arbitrary routing}
\label{secpi1rtng}
In Section \ref{secpi1}, we saw that when straight line routing (proposed in \cite{gupta}) is used, the per-node throughput is
$O\left(\frac{1}{n}\right)$. 
In this section, we show that this bound on the throughput also holds for a broad class of
routing schemes, i.e., Proposition \ref{propmain1} holds even when the assumption of straight-line routing is relaxed.
In this section, our goal is to investigate if a higher per-node
throughput can be achieved if we choose routing paths intelligently, i.e., not necessarily straight line 
paths between source-destination pairs. We only consider those routing paths in which
multi-hop communication takes place through packet-relaying between adjacent cells,
and there are no loops in the routing path.
The first condition is to ensure that communication is not possible between two nodes if they do not belong to adjacent cells,
while the second condition ensures that a packet does not visit a cell more than once on its path from source to destination.
Note that even with this restriction, we can account for a vast range of routing schemes.

Let us assume that a packet of connection $i$ takes a certain path (not necessarily straight line
path) from the source node to the destination node. Let ${\hat{L}}_{i}$ be the length of this path, i.e.,
the actual distance traveled by the packet along the surface of $S^2$. Then we have the following Lemma which is essentially a
generalization of Lemma \ref{lemmaHh}. In this Lemma, we show that at least a fixed fraction of hops for any routing scheme are over
a distance of $t\rho_n$ or longer for a fixed $t>0$. 

\begin{lemma}
\label{lemmartng}
Fix $W=40$, and $t=0.01$.
Assume {\em any} arbitrary routing path such that 
\begin{description}
\item[R1] Each hop is between nodes of two adjacent cells,
\item[R2] There are no loops in the routing path.
\end{description}
Then, under scheduling policy $\pi_1$, and R1 and R2, we cannot have $W$ consecutive hops of distance $t \rho_n$ or {\em less}. In other words, for any routing scheme
that satisfies R1 and R2, at least one out of every $W$ consecutive hops is at least $t \rho_n$ long.
\end{lemma}
\begin{proof}
See Appendix B.
\end{proof}

In the next Lemma, we show that except for a small fraction of hops, all other
hops have at least one concurrent transmission within a
distance of $(M+8)\rho_n$ of the receiver, and this fraction can be made arbitrarily small by choosing $M$ large enough.
\begin{lemma}
\label{lemmaMgen}
Fix $M>16$. Consider an arbitrary routing path between source and destination nodes such that R1 and
R2 are satisfied.
Under scheduling policy $\pi_1$, let $N_i$ be the number of hops of connection $i$ such that
there is no concurrent transmission within a circle of radius $(M+8) \rho_n$ around
the receivers of those hops. To avoid tedious boundary conditions, let us not count the hops containing the source
and the destination nodes in $N_i$. Then,
\begin{equation*}
N_i \leq \frac{{\hat{L}}_i}{\rho_n} \left( \frac{2(1 + c_1)}{M}\right)\mbox{,}
\end{equation*}
where $\hat{L}_i$ is the length of the routing path and $c_1$ is the constant from Lemma 4.4 in \cite{gupta}.
\end{lemma}
\begin{proof}
See Appendix B.
\end{proof}

Analogous to Proposition \ref{propphi} in Section \ref{secpi1}, we have the following result for the arbitrary routing model.

\begin{proposition}
\label{propphinew}
There exists a fixed constant $M_0$, that does not depend on $n$, such that for at least
$\hat{L}_i/640\rho_n$ hops of connection $i$, the SINR is less than a fixed constant $\beta_1$ given by
\begin{equation*}
\beta_1    = 100^\alpha \left( {M_0 + 8} \right)^{\alpha}.
\end{equation*}
Since the SINR is upper bounded by a fixed constant $\beta_1$, the
probability of successful packet reception is also upper bounded by a fixed
constant $\phi(\beta_1) < 1$.
\end{proposition}
\begin{proof}
See Appendix B.
\end{proof}

We can now use the above Proposition, and proceed as in Section \ref{secpi1} to determine the end-to-end throughput of a
connection. For this, we note that at least $\hat{L}_i/640\rho_n$ hops of connection $i$ have an SINR that is upper bounded by a
constant $\beta_1$, i.e., a packet success probability that is upper bounded by a constant $\phi(\beta_1)$. 
Hence, replacing $L_i$ by ${\hat{L}}_i$ in \eqref{Lambdafirst} to \eqref{Lambda} we obtain the throughput a connection as
\begin{equation*}
{\bf E} [\Lambda_n ]         \leq \lambda_n \, {\bf E}_L \left[ \delta^{{\hat{L}}_i} \right].
\end{equation*}
Recall that $L_i$ denotes the length of the straight line joining the source-destination node pair of connection $i$,
and hence ${\hat{L}}_i \geq L_i$. Also, since $\delta \leq 1$, we get
\begin{equation*}
{\bf E} [\Lambda_n ] \leq \lambda_n \, {\bf E}_L \left[ \delta^{L_i} \right].
\end{equation*}
After this stage, we use \eqref{deltaL}, and proceed along the same
lines as Section \ref{secpi1}, and obtain the same asymptotic upper bound of $O\left(\frac{1}{n}\right)$
as in Propositions \ref{propmain1}. 
\begin{proposition}
\label{propmain2new}
For any choice of $\rho_n$ that satisfies the connectivity requirements
in \cite{guptaconn}, and for any routing scheme that
satisfies R1 and R2, the per-node throughput that can be achieved under
scheduling policy $\pi_1$ is $O(\frac{1}{n})$.
\end{proposition}
\endproof
Thus the asymptotic per-node throughput is
$O\left(\frac{1}{n}\right)$ even with arbitrary routing, demonstrating that the impact of cumulative
packet loss out-weighs intelligent routing.

\section*{Appendix B: Detailed Proofs}
\subsection*{Proof of Lemma \ref{lemmaH}}

Since each Voronoi cell is contained in a circle of radius $2 \rho_n$, the maximum distance between
a point in a given cell, and a point in a neighboring cell is $8 \rho_n$ (see Fig.~\ref{figH}). 
Thus, over each hop, the
maximum distance that a packet can cover is upper bounded by $8 \rho_n$. Hence the lower bound.
\begin{figure}
\epsfysize=3.25cm
\centerline{\epsfbox{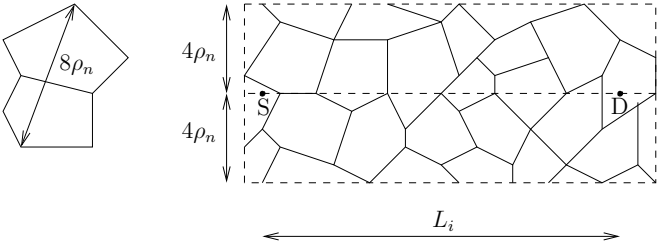}}
\caption{Bounding the number of hops $H_i$ for connection $i$.}
\label{figH}
\end{figure}

For the upper bound, if we look at a strip that is $4 \rho_n$ wide on either sides of $L_i$ 
(see Fig.~\ref{figH}), we observe that if a
cell is used as a relay cell, then it must lie {\em entirely}\, within this strip. This is because
each cell is completely contained inside a circle of diameter $4 \rho_n$. The area of this strip is
$8 \rho_n L_i$. Each cell contains a circle of radius $\rho_n$, and hence the area of each cell is at
least $\pi {\rho_n}^2 / 2$ (see Lemma \ref{lemmaEXP2} in Appendix). Thus, the maximum number of
cells that lie {\em entirely}\, within this strip
is upper bounded by $16 L_i / \pi \rho_n$. Hence the upper bound.

The lower bound is actually $\max \left( \frac{L_i}{8 \rho_n}, 1\right)$, but with
$L_i=\Theta(\sqrt{n})$, and $\rho_n=\Theta\left(\sqrt{\log{n}}\right)$, and we can ignore one. 
In any case, we require the upper bound, and not so much the lower bound in our subsequent analysis.
Also note, that for a more exact upper bound, we should consider a strip of length $L_i + 4 \rho_n + 4
\rho_n$ instead
of $L_i$, since the source and destination nodes could be located at the peripheries of their
respective cells. However, since $L_i$ is $\Theta \left(\sqrt{n}\right)$, and $\rho_n=\Theta\left(\sqrt{\log{n}}\right)$,
we can safely ignore the two terms corresponding to the sizes of the source and destination cells.
\endproof

\subsection*{Proof of Lemma \ref{lemmaHh}}

Since $h_i$ hops each cover less than $t \rho_n$ distance, the leftover distance, which is at least
$L_i - h_i t \rho_n$ has to be covered by the rest of the $H_i - h_i$ hops. Each of these $H_i - h_i$
hops can cover a distance of at the most $8 \rho_n$. Hence, we have
\begin{align*}
(H_i - h_i)\, 8 \rho_n &\geq L_i - h_i(t \rho_n) \\
\Rightarrow h_i &\leq H_i \left( \frac{8}{8-t} \right) - \frac{L_i}{\rho_n}
\left(\frac{1}{8-t}\right) \\
\Rightarrow H_i - h_i &\geq H_i \left(\frac{-t}{8-t}\right) + \frac{L_i}{\rho_n} \left(\frac{1}{8-t}\right)\\
 &\geq \frac{L_i}{\rho_n} \left(\frac{16}{\pi}\right) \left(\frac{-t}{8-t}\right) +
 \frac{L_i}{\rho_n}\left(\frac{1}{8-t}\right)\mbox{,}
\end{align*}
where we have used Lemma \ref{lemmaH} to upper bound $H_i$ on the right hand side in the last step.
Thus,
\begin{equation*}
H_i - h_i \geq \frac{L_i}{\rho_n}\left(\frac{1 - \frac{16 t}{\pi}}{8-t}\right).
\end{equation*}
\endproof

\subsection*{Proof of Lemma \ref{lemmaM}}
To avoid tedious boundary conditions, let us not count the hops containing the source
and the destination nodes in $N_i$. 
There are at the most $1+c_1$ colors, i.e., each cell
gets a transmission opportunity at least
once every $1+c_1$ time slots. Let us index these $1+c_1$ colors, and consider any one of these colors,
say color $j$. 
We would like to determine the number of color-$j$ cells of connection $i$ that are at least 
$M \rho_n$ distance away from other color-$j$ cells.

For this, let $U_i^j$ be the set of all the cells of color $j$ that a packet of connection $i$ traverses.
We refer to a circle as a ``surrounding circle of a cell'' if its {\em center lies within the cell}.
Consider all the surrounding
circles of radius $M \rho_n$ around each of the cells in set $U_i^j$. 
For each cell, there are many such circles, since the only requirement is that
the center of the circle lie within the cell. A subset $V_i^j$ of $U_i^j$ is then formed as follows.
A cell from set $U_i^j$ is added to set $V_i^j$ if {\em none} of
its surrounding circles of radius $M \rho_n$ contains  any other color-$j$ cell either partially or fully.
Thus every cell in set $V_i^j$ is at least $M \rho_n$ distance away from any other color-$j$ cell.

\begin{figure}
\epsfysize=3cm
\centerline{\epsfbox{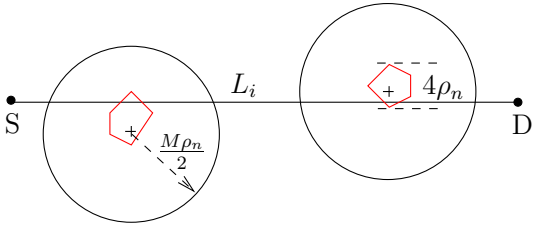}}
\caption{Bounding the number of hops $N_i^j$ of connection $i$ with color $j$ for which there is no interferer within
a distance $(M+8) \rho_n$ from the receiver.}
\label{figM}
\end{figure}

Hence, any surrounding circle of radius $\frac{M \rho_n}{2}$ around a cell in set $V_i^j$ does not 
overlap with any surrounding circle of radius $\frac{M \rho_n}{2}$ around another cell in set $V_i^j$.
Let there be $N_{i}^{j}$ cells in set $V_i^j$. 
The total length of the path is $L_i$ (referring to Fig.~\ref{figM}). We would like to determine the
maximum number of {\em non-overlapping surrounding}\, circles of radius $\frac{M \rho_n}{2}$ around color-$j$ cells that can be
accommodated along $L_i$. This number is an upper bound for $N_i^j$.
Each cell along $L_i$ can be enclosed inside a circle of diameter $4
\rho_n$. Referring to Fig.~\ref{figM}, maximal packing of the non-overlapping circles of radius $\frac{M \rho_n}{2}$ occurs when the
chord formed by the intersection of $L_i$ with these circles is of minimal length. 
If the diameter of the
circle encompassing the cell lies above (below) $L_i$, the cell can be shifted up (down) by at most $2\rho_n$ and still
intersect $L_i$.
Maximal packing occurs when each cell tangentially touches $L_i$, so that we select a
point in that cell $4\rho_n$ away from $L_i$, and draw a circle of radius $\frac{M\rho_n}{2}$ centered at that point.
The length of each chord formed by the intersection of $L_i$ and the surrounding circle of radius $\frac{M \rho_n}{2}$ is at
least $2\left( \sqrt{\left(\frac{M \rho_n}{2}\right)^2 - \left(4 \rho_n\right)^2} \right)$,
which in turn is at least $\frac{M \rho_n}{2}$ if $M> 9$. Since $N_{i}^{j}$ such chords are to 
be accommodated along the length of the path,
\begin{equation*}
N_{i}^{j} \frac{M \rho_n}{2} \leq L_i \Rightarrow N_{i}^{j}  \leq \frac{2 L_i}{M\rho_n}.
\end{equation*}
Therefore for all the colors,
\begin{align*}
N_i &= \sum\limits_{j=1}^{1+c_1} N_{i}^{j}
\leq\frac{L_i}{\rho_n}\left( \frac{2(1+c_1)}{M} \right).
\end{align*}
Thus except for these $N_i$ hops, each of the hops of connection $i$ has at least one cell that has a
simultaneous ongoing transmission
within a radius of $M \rho_n$ of the given cell. Or equivalently, within a radius of $(M+4+4)
\rho_n = (M+8)\rho_n$ of the {\em receivers}\, of these hops, there is at least one more
node that is transmitting simultaneously. This is because the maximum size of a cell is $4 \rho_n$.
And this proves the result.
\endproof

\subsection*{Proof of Proposition \ref{propphi}}
Let $A_i$ be the set of hops of connection $i$ for which the received signal is at the most $P
(t \rho_n)^{-\alpha}$. Then, given $\epsilon_1 > 0$, we can find $t>0$ small enough so that using
Lemma \ref{lemmaHh}, 
\begin{equation}
\label{nAi}
|A_i| \geq \frac{L_i}{\rho_n} \left( \frac{1}{8} - \epsilon_1 \right).
\end{equation}
Let $B_i$ be the set of hops of connection $i$ for which there is no simultaneous transmissions within a distance of 
$(M+8) \rho_n$ of the receiver.
Using Lemma \ref{lemmaM}, $|B_i| = N_i$, and given $\epsilon_2 >0$, 
we can choose $M$ large enough so that
\begin{align}
|B_i| &\leq \frac{L_i}{\rho_n} \left( \frac{2(1+c_1)}{M}\right) \leq \frac{L_i}{\rho_n}
\epsilon_2.
\label{nBi}			
\end{align}
Thus using \eqref{nAi} and \eqref{nBi},
\begin{align*}
|A_i \cap {B_i}^{c}| &\geq |A_i| - |B_i| \nonumber \\
& \geq \frac{L_i}{\rho_n} \left( \frac{1}{8} - \epsilon_1\right) - \frac{L_i}{\rho_n} \epsilon_2 \nonumber \\
& = \frac{L_i}{\rho_n} \left( \frac{1}{8} - \epsilon_1 - \epsilon_2\right).
\end{align*}
If we pick $\epsilon_1 = \epsilon_2 = 1/32$, and choose $t=t_0$ and $M=M_0$ corresponding to this choice of
$\epsilon_1$, $\epsilon_2$, then 
\begin{equation}
|A_i \cap {B_i}^{c}| \geq \frac{L_i}{16 \rho_n}.
\label{nAB}
\end{equation}
Note that $A_i \cap {B_i}^{c}$ is the set of hops over which the received signal is no more than $P
(t_0 \rho_n)^{-\alpha}$, {\it and} there is at least one simultaneous transmission within a distance of $(M_0
+ 8) \rho_n$ of the receiver. This in turn means that for these hops, the SINR is upper bounded by
\begin{align}
SINR &\leq \frac{P (t_0 \rho_n)^{-\alpha}}{P \left((M_0 + 8)\rho_n\right)^{-\alpha}} \nonumber \\
\label{sinr}
&= \left( \frac{M_0 + 8}{t_0} \right)^{\alpha} = \beta_0. 
\end{align}

\endproof

\subsection*{Proof of Proposition \ref{propinter}}
Let $\lambda_n $ be the rate in packets/slot at which every source node injects packets in the network.
Although the source node injects packets at a rate of $\lambda_n $ packets per slot, not all the
packets make it to the destination node. This is because, at each hop, the SINR is finite, and hence
there is a non-zero probability of the packet getting dropped. The actual end-to-end throughput of
the connection, denoted by $\Lambda_n $ is given by,
\begin{align}
\Lambda_n  = \lambda_n \, &\mbox{Prob} \left\{\mbox{packet is received
successfully over}\right. \nonumber \\
                              &\left.\mbox{all the hops of connection $i$}\right\} \nonumber \\
 = \lambda_n \, &\Pi_{j=1}^{H_i} \mbox{Prob} \left\{\mbox{packet is received
 successfully}\right. \nonumber \\
                 &\left.\mbox{over the $j$th hop of connection
								 $i$}\right\}\mbox{,}\nonumber 
\end{align}
where we assume that the interference and noise observed by a packet at each of the hops
are independent. We know from Proposition \ref{propphi} that, among the $H_i$ hops of connection $i$, at least $L_i/16 \rho_n$ hops have a probability of packet
success of no more than $\phi(\beta_0)$. Thus,
\begin{equation}
\Lambda_n  \leq \lambda_n \, \left\{\phi(\beta_0)\right\}^{\frac{L_i}{16\rho_n}}.
\label{Lambdafirst}
\end{equation}
Note that $L_i$ are i.i.d. random variables. Hence if we remove the conditioning on $L_i$ by taking
the expectation with respect to $L_i$, the end-to-end throughput ${\bf E} [\Lambda_n ]$ is,
\begin{align}
{\bf E} [\Lambda_n  ]&= {\bf E}_L [\Lambda_n ]  \\
          &\leq \lambda_n \, {\bf E}_L \left\{
					\left(\left\{\phi(\beta_0)\right\}^{\frac{1}{16\rho_n}}\right)^{L_i} \right\}	\\
          \label{Lambda}
          &= \lambda_n \, {\bf E}_L \left[ \delta^{L_i} \right].
\end{align}
where we have substituted $\delta = \left\{\phi(\beta_0)\right\}^{\frac{1}{16\rho_n}}$. 
Note that in determining the average end-to-end throughput ${\bf E} [\Lambda_n  ]$, we take expectations at two
levels; once to take into account the randomness due to the possibility of packet error on each link, and
once to take into account the randomness due to the locations of the source and destination nodes.
Also note that $0 < \delta < 1$. Since $L_i$ is a line connecting two points
picked at random on the surface of $S^2$, we can show that (See Lemma \ref{lemmaEXP} in Appendix).
\begin{equation}
\label{deltaL}
{\bf E}_L \left[ \delta^{L_i} \right] = \frac{2\pi\left(1 +
\delta^{\frac{\sqrt{\pi n}}{2}}\right)}{4\pi + n(\log \delta)^2}.
\end{equation}
Using \eqref{Lambda} and \eqref{deltaL},
\begin{align*}
{\bf E} [\Lambda_n ] &\leq \lambda_n \,  \frac{2\pi\left(1 + \delta^{\frac{\sqrt{\pi n}}{2}}\right)}{(4\pi + n(\log \delta)^2)} \nonumber \\
          &=   \lambda_n \,   \frac{2\pi\left(1 +
					\left\{\phi(\beta_0)\right\}^{\frac{\sqrt{\pi n}}{32\rho_n}}\right)}{\left(4\pi
					+ n\left(\frac{1}{16\rho_n}
					\log {\phi(\beta_0)}\right)^2\right)} \nonumber \\
          &=  \lambda_n \,    \frac{512\pi{\rho_n}^2\left(1 +
					\left\{\phi(\beta_0)\right\}^{\frac{\sqrt{\pi n}}{32\rho_n}}\right)}{\left(1024\pi{\rho_n}^2
					+ n(\log {\phi(\beta_0)})^2\right)} \nonumber \\
          &< \lambda_n \,  \frac{1024\pi{\rho_n}^2 }{\left( 1024\pi{\rho_n}^{2} + 
					n(\log {\phi(\beta_0)})^2\right)}\mbox{, } \mbox{ since $\phi(\beta_0) < 1$.} \nonumber \\
          &< \lambda_n \,  \frac{1024\pi{\rho_n}^2 }{n(\log
					{\phi(\beta_0)})^2}\nonumber \\
					&\leq \lambda_n \, \frac{c_0 \log{n}}{n}, \hspace{0.2cm} \text{since $\rho_n=\Theta\left(\sqrt{\log{n}}\right)$.}
\end{align*}

\endproof

\subsection*{Proof of Lemma \ref{lemmartng}}
\begin{figure}
\epsfysize=4cm
\centerline{\epsfbox{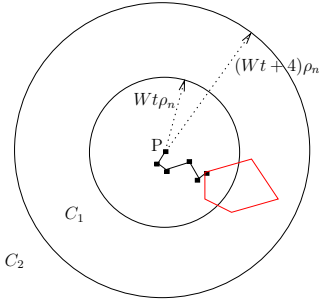}}
\caption{Bounding the number of hops of length $t \rho_n$ or less.}
\label{figrtng}
\end{figure}
By way of contradiction, assume that there exists some routing scheme that
satisfies the connectivity
requirements in \cite{guptaconn}, along with R1 and R2, and has $W$ consecutive hops, each of length $t \rho_n$ or less. 
Referring to Fig.~\ref{figrtng}, let point P be the initial location of the packet, and assume that the packet traverses $W$ hops from
this position. Since each of these hops is less than $t\rho_n$, the final location of the packet after traversing these $W$ hops is
inside a circle of radius $W t \rho_n$ centered at P. Let $C_1$ be this circle.
Clearly, every cell that the packet traversed during these $W$ hops must intersect circle $C_1$. Otherwise, the packet could not have
traversed that cell. We can easily find an upper bound on the number of cells that intersect $C_1$ as follows.
Consider another circle $C_2$ of radius $(W t+4)\rho_n$ centered at P. If a cell intersects $C_1$, it should lie {\em completely} inside
circle $C_2$. This is because the distance between any two points of a cell cannot exceed $4\rho_n$, and if a cell has at least one point
inside $C_1$, its other point cannot be more than $4\rho_n$ farther from the circumference of circle $C_1$.
The number of cells that lie completely inside $C_2$, denoted by $n_2$ is upper bounded as follows
\begin{align}
n_2 &\leq \frac{\pi \left((W t+4)\rho_n\right)^2}{\left(\frac{\pi {\rho_n}^2}{2}\right)} \nonumber \\
    &= 2 (W t+4)^2 \text{, }
\end{align}
where we have used Lemma \ref{lemmaEXP2} in the first step. For $W=40$ and $t=0.01$, we have $n_2 \leq 38.72$.  
Since 40 hops require at least 41 unique cells, and there are only up to 38 unique cells that the packet
can visit, we clearly have a contradiction.
Thus, there does not exist a routing scheme that satisfies, R1 and R2, and that has $W$ consecutive hops of $t \rho_n$ 
or less, where $W=40$, and $t=0.01$.
\endproof
Note that the arguments in the above Lemma are easier to understand in case of simple tessellations such as triangular, square, hexagonal, etc.
For example, consider a square tessellation (square-shaped cells), such that each side of the cell is $2\rho_n$ long. Thus, each cell contains a circle of
radius $\rho_n$, and is contained in a circle of radius $2\rho_n$. In this case, it is easy to see that circle $C_1$ defined in Lemma \ref{lemmartng}
cannot intersect more than 4 circles, and this happens in the neighborhood of one of the vertices of the tessellation. 
Consequently, for a square tessellation, we cannot have 4 consecutive hops, each of length $0.01 \rho_n$ or less.
In fact, the upper bound could be made tighter by noting that there cannot be more than 4 consecutive
hops that are strictly less than $0.5 \rho_n$.
Similarly, for a hexagonal tessellation such that each cell has sides of length $2\rho_n$,
we cannot have 3 consecutive 
hops, each of length strictly less than $0.5 \rho_n$. Thus the value of $W$ is much smaller, and the value of 
$t$ is much larger for simple tessellations. However when we consider general tessellations that
satisfy the connectivity requirements, we note that the tessellation need not be as well-behaved as square or hexagonal. Yet, we
can obtain an upper bound on the maximum number of successive hops of arbitrarily small lengths. To do
this, we need to choose a fairly large (but finite and fixed) $W$, and a fairly small (but positive
and fixed) $t$.

\subsection*{Proof of Lemma \ref{lemmaMgen}}
\begin{figure}
\epsfxsize=8cm
\centerline{\epsfbox{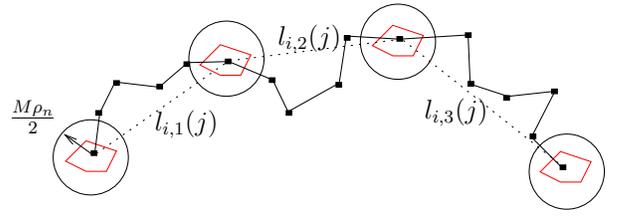}}
\caption{Arbitrary routing path.}
\label{figMgen}
\end{figure}
We use an approach identical to the one used to prove Lemma \ref{lemmaM}. 
Let $U_i^j$ be the set of all the cells of color $j$ that a packet of connection $i$ traverses.
Consider all the surrounding
circles of radius $M \rho_n$ around each of the cells in set $U_i^j$. 
For each cell, there are many such circles, since the only requirement is that
the center of the circle lie within the cell. A subset $V_i^j$ of $U_i^j$ is then formed as follows.
A cell from set $U_i^j$ is added to set $V_i^j$ if {\em none} of
its surrounding circles of radius $M \rho_n$ contains  any other color-$j$ cell either partially or fully.
Thus every cell in set $V_i^j$ is at least $M \rho_n$ distance away from any other color-$j$ cell.
Hence, any surrounding circle of radius $\frac{M \rho_n}{2}$ around a cell in set $V_i^j$ does not 
overlap with any surrounding circle of radius $\frac{M \rho_n}{2}$ around another cell in set $V_i^j$.
Let there be $N_{i}^{j}$ cells in set $V_i^j$,  and let $N_i$ be the sum of the number of cells in
$V_i^j$ for all $j$. We would like to determine an upper bound on $N_i$.

Referring to Fig.~\ref{figMgen}, consider the path taken by a packet of connection $i$ (solid line in
the figure). The path is piece-wise linear, with the locations of the relaying node of each hop as points of discontinuity.
Let $\hat{L}_i$ be the length of this routing path.
Let us rearrange the cells in set $V_i^j$ {\em in the
order in which a packet of connection $i$ traverses them}, 
and join them in this order (dotted lines in Fig.~\ref{figMgen}).
Let $l_{i,k}(j)$ be the length of the $k$th dotted line segment.
Define ${\hat{L}}_i^{'}(j)$ to be the total length of the dotted path. Then we have
\begin{equation}
{\hat{L}}_i^{'}(j) = \sum_{k} l_{i,k}(j) \leq {\hat{L}}_i.
\label{lprime}
\end{equation}
\begin{figure}
\epsfxsize=8cm
\centerline{\epsfbox{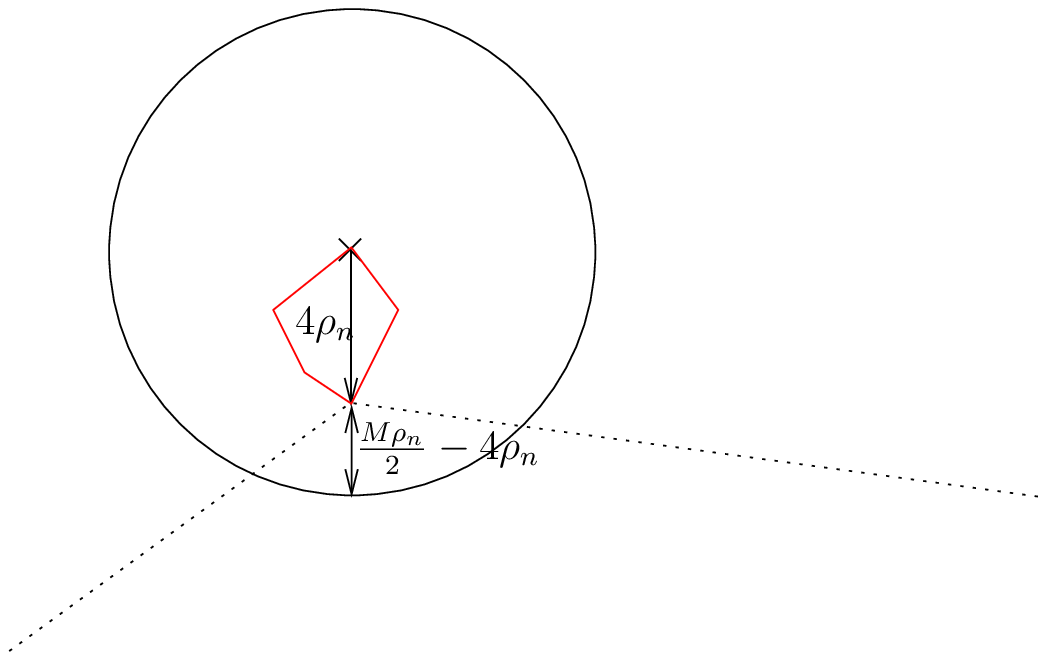}}
\caption{Bounding the number of hops $N_i$ of connection $i$ for which there is no interferer within
a distance $(M+8) \rho_n$ from the receiver with arbitrary (not straight line) routing.}
\label{figMgenlin}
\end{figure}
Consider a typical cell in set $V_i^j$ as shown in Fig.~\ref{figMgenlin}. 
As in the proof of Lemma \ref{lemmaM}, we first
determine the length of the minimum possible section of the routing path (in this case the dotted path) that lies inside the
surrounding circle. Since we know the total length of the dotted path, this gives us an upper bound on the number surrounding circles of
radius $\frac{M \rho_n}{2}$ that can be accommodated along the dotted path.

Referring to Fig.~\ref{figMgenlin}, the dotted path has its smallest section inside the surrounding circle of radius $\frac{M
\rho_n}{2}$
when the relaying node of the
cell, i.e., the point of discontinuity of the dotted path is located as close as possible to the circumference of the circle.
However, since this point belongs to the cell, it cannot be farther than $4\rho_n$ from any other point in the cell. But the center
of the surrounding circle must also lie inside the cell. Hence the point of discontinuity can be at the most $4\rho_n$ away from the
center of the circle. Hence the each portion of the
dotted path that lies within the circle can be no shorter than $\frac{M \rho_n}{2} - 4 \rho_n$. Thus, for each cell in $V_i^j$, at
least $(M-8)\rho_n$ portion of the dotted path must lie inside the surrounding circle. Let
$M-8>\frac{M}{2}$, i.e., $M>16$. Then the maximum number of such surrounding circles that can be accommodated along the dotted path,
i.e., the maximum number of cells in set $V_i^j$ is:
\begin{equation*}
N_{i}^{j} \frac{M \rho_n}{2} \leq \hat{L}_i^{'} \Rightarrow N_{i}^{j}  \leq \frac{2
\hat{L}_i^{'}}{M\rho_n}.
\end{equation*}
Therefore for all the colors,
\begin{align*}
N_i &= \sum\limits_{j=1}^{1+c_1} N_{i}^{j} \\
&\leq\frac{\hat{L}_i^{'}}{\rho_n}\left( \frac{2(1+c_1)}{M} \right) 
\leq\frac{\hat{L}_i}{\rho_n}\left( \frac{2(1+c_1)}{M} \right)\text{using \eqref{lprime}.}
\end{align*}
Thus except for these $N_i$ hops, each of the hops of connection $i$ has at least one cell that has a
simultaneous on-going transmission
within a radius of $M \rho_n$ of the given cell. Or equivalently, within a radius of $(M+4+4)
\rho_n = (M+8)\rho_n$ of the {\em receivers}\, of these hops, there is at least one more
node that is transmitting simultaneously. This is because the maximum size of a cell is $4 \rho_n$.
And this proves the result.
\endproof

\subsection*{Proof of Proposition \ref{propphinew}}
For any routing path, the number of hops is lower bounded by $\frac{\hat{L}_i}{8 \rho_n}$, since each hop can traverse a distance of
at the most $8 \rho_n$. Lemma \ref{lemmartng} shows that at least $\frac{1}{40}$th of these hops, i.e., at least
$\frac{\hat{L}_i}{320 \rho_n}$ hops of connection $i$ have transmitter-receiver pair separated by $0.01 \rho_n$ or more. 
As in Section \ref{secpi1}, let $A_i$ be the set of hops of connection $i$ for which the received signal is at the most $P
(0.01 \rho_n)^{-\alpha}$. Then,
\begin{equation}
\label{newnAi}
|A_i| \geq \frac{\hat{L}_i}{320 \rho_n}.
\end{equation}
Let $B_i$ be the set of hops of connection $i$ for which there is no simultaneous transmissions within a distance of 
$(M+8) \rho_n$ of the receiver.
Using Lemma \ref{lemmaMgen}, given $\epsilon >0$, 
we can choose $M$ large enough so that
\begin{align}
|B_i| &\leq \frac{\hat{L}_i}{\rho_n} \left( \frac{2(1+c_1)}{M}\right)
      \leq \frac{\hat{L}_i}{\rho_n} \epsilon.
\label{newnBi}			
\end{align}
Then using \eqref{newnAi} and \eqref{newnBi},
\begin{align*}
|A_i \cap {B_i}^{c}| &\geq |A_i| - |B_i| \nonumber \\
& \geq \frac{\hat{L}_i}{\rho_n} \left( \frac{1}{320}\right) - \frac{\hat{L}_i}{\rho_n} \epsilon \nonumber \\
& = \frac{\hat{L}_i}{\rho_n} \left( \frac{1}{320} - \epsilon\right).
\end{align*}
If we pick $\epsilon =  1/640$, and choose $M=M_0$ corresponding to this choice of
$\epsilon$,
\begin{equation}
|A_i \cap {B_i}^{c}| \geq \frac{\hat{L}_i}{640 \rho_n}.
\label{newnAB}
\end{equation}
Note that $A_i \cap {B_i}^{c}$ is the set of hops over which the received signal is no more than $P
(0.01 \rho_n)^{-\alpha}$, {\it and} there is at least one simultaneous transmission within a distance of $(M_0
+ 8) \rho_n$ of the receiver. This in turn means that for these hops, the SINR is upper bounded by
\begin{align}
\label{newsinr}
SINR &\leq \frac{P (0.01 \rho_n)^{-\alpha}}{P \left((M_0 + 8)\rho_n\right)^{-\alpha}} \\
\label{newsinr}
&= 100^\alpha \left( {M_0 + 8} \right)^{\alpha}
= \beta_1 .\nonumber
\end{align}
The above upper bound on SINR may seem large due to the $100^\alpha$ term. However, this term
appears because we picked up a conservative hop length of $0.01 \rho_n$ in proving Lemma
\ref{lemmartng} to account for all possible tessellations. For simple tessellations such as square and
hexagonal, we can use a minimum hop length of $\rho_n$, and obtain a tighter bound on
SINR. However, as far as the asymptotic nature of the solution is
concerned, the values of these constants do not matter.
\endproof

\subsection*{Proof of Lemma \ref{unionbnd}}
\begin{align*}
&P\left(\left\{\text{End-to-end packet loss on connection $j$}\right\}\right)\\
&=P\left(\cup_{i}\left\{\text{Packet loss on link $i$ of connection $j$}\right\}\right)\\
&\leq\sum\limits_{i} P\left(\left\{\text{Packet loss on link $i$ of connection $j$}\right\}\right)\\
&\leq\sum\limits_{i} \frac{\sqrt{\pi}\rho_n \epsilon}{8\sqrt{n}}\\
&=\frac{\sqrt{\pi}\rho_n \epsilon}{8\sqrt{n}}\cdot H_j\\
&\leq\frac{\sqrt{\pi}\rho_n \epsilon}{8\sqrt{n}}\cdot \frac{16}{\pi}\frac{L_j}{\rho_n},
\hspace{0.6cm} \text{using Lemma \ref{lemmaH}.}\\
&\leq \epsilon,
\end{align*}
since $L_j$ is upper bounded by $\frac{\sqrt{\pi n}}{2}$, half the circumference of the sphere of area $n$.
\endproof

\subsection*{Proof of Lemma \ref{schedreq}}
Assuming that the interference observed at a node
has tail error probability, $P_e$ that falls of exponentially with SINR, 
for a given $\epsilon>0$, using Lemma \ref{unionbnd} we require a
minimum SINR of $\beta_n$ for all the links such that

\begin{align*}
{P_e}(\beta_n) &\leq
\frac{\sqrt{\pi}}{8\sqrt{n}}\rho_n \epsilon 
= \frac{5\sqrt{\pi}}{4} \sqrt{\frac{\log{n}}{n}} \epsilon.
\end{align*}

Since for large $x$, $P_e (x)$ can be upper bounded by $e^{-x}$,
we have for $x=\beta_n={\log{n}}$,
\begin{align*}
P_e({\log{n}}) &= e^{-\log{n}} \\
&= \frac{1}{n} \\
&\leq \frac{5\sqrt{\pi}}{4} \sqrt{\frac{\log{n}}{n}} \epsilon, \hspace{0.5cm} \text{for $n$ large.}
\end{align*}
Thus, if a scheduling policy ensures an SINR of at least
$\beta_n={\log{n}}$, then an end-to-end packet loss probability of
$\epsilon$ can be guaranteed for a given $\epsilon>0$ using Lemma \ref{unionbnd}.
\endproof

\subsection*{Proof of Lemma \ref{pisched}}
Consider an arbitrary cell $C_i$. Fix the origin to be the center of this cell, i.e.,
the common center of the circles of radii $20\sqrt{{\log{n}}}$ and $10\sqrt{{\log{n}}}$ that
respectively contain, and are contained in the cell. Consider a ring of
thickness $dr$ at a distance $r$ from the center. Since the area of each cell is
lower bounded by ${50\pi\log{n}}$, the number of cells within this ring can be no
more than $\frac{r dr}{25\log{n}}$. Even if all these cells are transmitting
concurrently, the total interference received at the center of $C_i$, and denoted by
$I_i$ is upper bounded as follows.
\begin{align*}
I_i &\leq \int\limits_{R_n}^{\infty} \frac{ r dr}{25\log{n}} \cdot P r^{-\alpha}
= \frac{ P R_n^{-\alpha+2}}{25\log{n} (\alpha-2)}.
\end{align*}
Since each cell is contained inside a circle of radius $20\sqrt{{\log{n}}{}}$, the maximum
separation between a transmitter and a receiver that lie in adjacent cells is
$80\sqrt{{\log{n}}{}}$. Hence the SINR in cell $C_i$ is greater than ${\log{n}}$ if $R_n$
satisfies the following.
\begin{align*}
\frac{P\cdot\left(80\sqrt{{\log{n}}{}}\right)^{-\alpha}}{\frac{P
R_n^{-\alpha+2}}{(\alpha-2)25\log{n}}} &\geq {\log{n}}\\
\Rightarrow R_n &\geq 80
\left(\frac{256}{(\alpha-2)}\right)^{\frac{1}{\alpha-2}}\left({\log{n}}\right)^{\frac{\alpha}{2(\alpha-2)}}.
\end{align*}
The above ensures that SINR at the center of cell $C_i$ is greater than
${\log{n}}$. To ensure an SINR of at least ${\log{n}}$ everywhere in $C_i$, we note
that the farthest distance to the cell edge from center is $20\sqrt{{\log{n}}{}}$. Thus, we require
\begin{align*}
R_n &\geq 
20\sqrt{{\log{n}}}  +
80\left(\frac{256}{(\alpha-2)}\right)^{\frac{1}{(\alpha-2)}}\cdot\left(\log{n}\right)^{\frac{\alpha}{2(\alpha-2)}} \\
&=20\sqrt{{\log{n}}} \left(1  +
4\left(\frac{256{{{\log{n}}}}}{(\alpha-2)}\right)^{\frac{1}{(\alpha-2)}}\right)
\end{align*}
\endproof

\begin{lemma}
\label{lemmaEXP}
If we pick two points randomly and uniformly on the surface of sphere $S^2$ of area $n$, and connect them by a
line (drawn along the surface of $S^2$), and if ${\bf L}$ is the random variable corresponding to the
length of this line, then for $0 < \delta < 1$,
\begin{equation*}
{\bf E} \left[ \delta^{\bf L} \right] = \frac{2\pi\left(1 +
\delta^{\frac{\sqrt{\pi n}}{2}}\right)}{4\pi + n(\log \delta)^2}.
\end{equation*}
\end{lemma}
\vspace{0.2cm}

\begin{proof}
Since both the points are picked randomly and uniformly on the surface of $S^2$, without loss of
generality, assume that one point is located at the north pole of the sphere. Then, we can find the
probability that the other point is located within a distance of $l$ from the north
pole. This is the same as the event $A= \left\{{\bf L} \leq l \right\}$, and its probability 
equals the ratio of the
area of the shaded region in Fig.~\ref{figEXP} to the area of the sphere, which is one.
Note that all the distances are measured along the surface of the sphere. 
\begin{figure}
\epsfysize=4cm
\centerline{\epsfbox{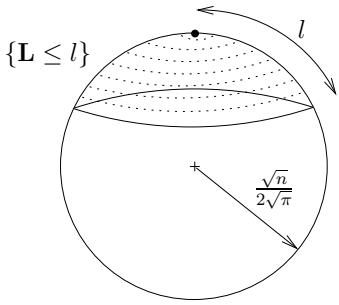}}
\caption{Probability distribution of ${\bf L}$.}
\label{figEXP}
\end{figure}
We have,
\begin{equation}
\mbox{Area}\left\{\mbox{Shaded region}\right\}
= \frac{n}{2} \left( 1 - \cos{\left(\frac{2\sqrt{\pi}l}{\sqrt{n}}\right)}\right).
\label{area}
\end{equation}
Therefore the distribution function of ${\bf L}$, i.e., $F_{L} (l)$ is
\begin{equation*}
F_{L} (l) = \frac{1}{2}\left(1 - \cos{\left(\frac{2\sqrt{\pi}l}{\sqrt{n}}\right)}\right) \mbox{.}
\end{equation*}
Note that $0 \leq {\bf L} \leq \sqrt{\pi n}/2$. 
The density function of ${\bf L}$ is given by,
\begin{equation*}
f_{L} (l) = \frac{\sqrt{\pi}}{\sqrt{n}}\sin{\left(\frac{2\sqrt{\pi}l}{{n}}\right)} \mbox{. }
\end{equation*}
Now we can compute ${\bf E}[{\delta}^{\bf L}]$ as follows.
\begin{align}
{\bf E}\left[{\delta}^{\bf L}\right] 
&= \int\limits_{0}^{\frac{\sqrt{\pi n}}{2}}  {\delta}^{l}
\frac{\sqrt{\pi}}{\sqrt{n}}\sin{\left(\frac{2\sqrt{\pi}l}{\sqrt{n}}\right)} dl \nonumber \\
&= \frac{1}{2}\int\limits_{0}^{\pi}  {\gamma}^{x} \sin{x} dx\mbox{, } \quad \mbox{ where $\gamma =
{\delta}^{\frac{\sqrt{n}}{2\sqrt{\pi}}}$} \nonumber \\
&= \frac{1}{2} I.
\label{EdeltaL}
\end{align}
Solving for $I$ using integration by parts,
\begin{equation*}
I=\frac{1 + {\gamma}^{\pi}}{1 + (\log{\gamma})^2}.
\end{equation*}
Replacing $\gamma$ by ${\delta}^{\frac{1}{2\sqrt{\pi}}}$, and substituting the above in \eqref{EdeltaL},
\begin{equation}
{\bf E}\left[{\delta}^{\bf L}\right]  = \frac{2\pi\left(1 + {\delta}^{\frac{\sqrt{\pi
n}}{2}}\right)}{4\pi + n(\log{\delta})^2}.
\end{equation}
\end{proof}
\vspace{0.2cm}
Using \eqref{area} from the above Lemma, we also have the following lemma
for the area of a disk of radius $\rho_n$ on $S^2$.
\begin{lemma}
\label{lemmaEXP2}
If $\rho_n$ is the radius of a disk (measured along the surface of $S^2$),
then
\begin{equation}
\frac{\pi {\rho_n}^2}{2} \leq a_n \leq {\pi}{\rho_n}^2 \mbox{.}
\label{area2}
\end{equation}
\end{lemma}
\vspace{0.2cm}

\begin{proof}
Using \eqref{area} with $l=\rho_n$ gives
\begin{align*}
a_n &= \frac{n}{2}\left(1-\cos{\left(\frac{2\sqrt{\pi}\rho_n}{\sqrt{n}}\right)}\right) \\
    &= \sin^2{\left(\sqrt{\frac{\pi}{n}}\rho_n\right)}.
\end{align*}
For $\theta \leq {\pi}/4$,
\begin{equation*}
\frac{\theta}{\sqrt{2}} \leq \sin{\theta} \leq \theta.
\end{equation*}
From this, \eqref{area2} follows if $n$ is large enough so that $\rho_n \leq \sqrt{\pi n}/4$.
The latter is always true, since the $ \sqrt{\pi n}/4$ is the radius of a hemisphere on $S^2$ (measured along the
surface of $S^2$), and hence the radius of the largest disk on $S^2$.
\end{proof}

\end{document}